\documentclass[submission,copyright,creativecommons]{eptcs}
\usepackage[utf8x]{inputenc}
\usepackage{array}
\usepackage{amsthm}
\usepackage{amsfonts}
\usepackage{amsmath}    
\usepackage{graphicx}   
\usepackage{verbatim}   
\usepackage{color}      
\usepackage{subfigure}  
\usepackage{hyperref}   
\usepackage{amssymb}
\usepackage{tikz}
\usetikzlibrary{matrix,arrows,decorations.pathmorphing,shapes}
\usepackage{amssymb} 
\usepackage{latexsym}
\usepackage{epsfig}
\usepackage{stackrel} 
\usepackage{stmaryrd}
\usepackage{xspace}
\usepackage{graphicx}
\usepackage{xy}
\usepackage{color}
\usepackage{pdfsync}
\usepackage{enumerate}
\usepackage{atbeginend}

\input{growbox.sty}

\newcommand{\ket}[1]{\left| #1 \right>} 
\newcommand{\bra}[1]{\left< #1 \right|} 

\newcommand{\ruleT}{\ensuremath{\mathbf{T}}}
\newcommand{\ruleS}{\ensuremath{\mathbf{S}}}
\newcommand{\ruleSi}{\ensuremath{\mathbf{S1}}}
\newcommand{\ruleSii}{\ensuremath{\mathbf{S2}}}

\newcommand{\ruleBp}{\ensuremath{\mathbf{B'}}}
\newcommand{\ruleBi}{\ensuremath{\mathbf{B1}}}
\newcommand{\ruleBii}{\ensuremath{\mathbf{B2}}}

\newcommand{\ruleKi}{\ensuremath{\mathbf{K1}}}
\newcommand{\ruleKii}{\ensuremath{\mathbf{K2}}}
\newcommand{\ruleC}{\ensuremath{\mathbf{C}}}
\newcommand{\ruleDi}{\ensuremath{\mathbf{D1}}}
\newcommand{\ruleDii}{\ensuremath{\mathbf{D2}}}
\newcommand{\ruleA}{\ensuremath{\mathbf{A}}}

\newcommand{\ruleE}{\ensuremath{\mathbf{E}}}

\newcommand{\eruleT}{\ensuremath{(\ruleT)}\xspace}
\newcommand{\eruleS}{\ensuremath{(\ruleS)}\xspace}
\newcommand{\eruleSi}{\ensuremath{(\ruleSi)}\xspace}
\newcommand{\eruleSii}{\ensuremath{(\ruleSii)}\xspace}

\newcommand{\eruleBp}{\ensuremath{(\ruleBp)}\xspace}
\newcommand{\eruleBi}{\ensuremath{(\ruleBi)}\xspace}
\newcommand{\eruleBii}{\ensuremath{(\ruleBii)}\xspace}

\newcommand{\eruleKi}{\ensuremath{(\ruleKi)}\xspace}
\newcommand{\eruleKii}{\ensuremath{(\ruleKii)}\xspace}
\newcommand{\eruleC}{\ensuremath{(\ruleC)}\xspace}
\newcommand{\eruleDi}{\ensuremath{(\ruleDi)}\xspace}
\newcommand{\eruleDii}{\ensuremath{(\ruleDii)}\xspace}
\newcommand{\eruleA}{\ensuremath{(\ruleA)}\xspace}

\newcommand{\eruleE}{\ensuremath{(\ruleE)}\xspace}

\newcommand{\zxcalculus}{\textsc{zx}-calculus\xspace}

\def\bc{\begin{center}}
\def\ec{\end{center}}

\xyoption{all}

\usepackage{url}

\newcommand{\bit}{\begin{itemize}}
\newcommand{\eit}{\end{itemize}\par\noindent}
\newcommand{\ben}{\begin{enumerate}}
\newcommand{\een}{\end{enumerate}\par\noindent}
\newcommand{\beq}{\begin{equation}}
\newcommand{\beqq}{\begin{equation}\hspace{-1.2cm}}
\newcommand{\eeq}{\end{equation}\par\noindent}
\newcommand{\beqx}{\begin{equation*}}
\newcommand{\beqqx}{\begin{equation*}\hspace{-1.2cm}}
\newcommand{\eeqx}{\end{equation*}\par\noindent}
\newcommand{\beqa}{\begin{eqnarray*}}
\newcommand{\eeqa}{\end{eqnarray*}\par\noindent}
\newcommand{\beqn}{\begin{eqnarray}}
\newcommand{\eeqn}{\end{eqnarray}\par\noindent}

\def\PICT{\begin{center}{\Large Picture:} }
\def\EPICT{\end{center}}

\newcommand{\inlinegraphic}[2]{
  \dimendef\grafheight=255\dimendef\grafvshift=254
  \grafheight=#1
  \grafvshift=-0.5\grafheight
  \advance\grafvshift by 0.5ex
  \raisebox{\grafvshift}{\includegraphics[height=\grafheight]{Fig/#2}\xspace}
}

\newcommand{\ninlinegraphic}[2][1.0]{
  \dimendef\grafheight=255\dimendef\grafvshift=254
  \setbox0 = \hbox{\scalebox{#1}{\includegraphics{Fig/#2}}}
  \grafheight=\the\ht0
  \grafvshift=-0.5\grafheight
  \advance\grafvshift by 0.5ex
  \raisebox{\grafvshift}{\includegraphics[height=\grafheight]{Fig/#2}\xspace}
}

\graphicspath{{./}{Fig/}{Plots/}}

\newcounter{counter}

\newtheorem{lemma}[counter]{Lemma}
\newtheorem{corollary}[counter]{Corollary}
\newtheorem{definition}{Definition} 

\providecommand{\event}{QPL 2011}
\title{Superdense Coding with GHZ and Quantum Key Distribution with $W$ in the \zxcalculus }
\author{Anne Hillebrand\\
 \institute{Oxford University Computing Laboratory\\ Oxford, UK}
\email{anne.hillebrand@maths.ox.ac.uk}
}
\date{\today}
\def\titlerunning{SDC and QKD in the \zxcalculus}
\def\authorrunning{A. M. Hillebrand}

\begin{document}

\maketitle

\begin{abstract}
Quantum entanglement is a key resource in many quantum protocols, such as quantum teleportation and quantum cryptography. Yet entanglement makes protocols presented in Dirac notation difficult to verify. This is why Coecke and Duncan have introduced a diagrammatic language for quantum protocols, called the \zxcalculus \cite{intro3}. This diagrammatic notation is both intuitive and formally rigorous. It is a simple, graphical, high level language that emphasises the composition of systems and naturally captures the essentials of quantum mechanics. In the author's MSc thesis \cite{ik} it has been shown for over 25 quantum protocols that the \zxcalculus provides a relatively easy and more intuitive presentation. Moreover, the author embarked on the task to apply categorical quantum mechanics on quantum security; earlier works did not touch anything but Bennett and Brassard's quantum key distribution protocol, BB84. Two of the protocols in \cite{ik}, namely superdense coding with the Greenberger-Horne-Zeilinger state and quantum key distribution with the $W$-state, will be presented in this paper. 
\end{abstract}

\section{Introduction}

Quantum protocols are usually described in Dirac notation. Though such a presentation is adequate, it is low-level and therefore not a very intuitive formalism. The passage to a high level language was realized in \cite{AbramskyCoecke}, by relying on the compositional structure of monoidal categories. Corresponding presentations result in the form of quantum picturalism in \cite{ kindergarten,intro2, intro3,  intro1}, which relies on the diagrammatic presentation of symmetric monoidal categories, tracing back to Penrose \cite{Penrose, JoyalStreet,Selinger}.

This diagrammatic notation is both intuitive and formally rigorous. It is a simple, graphical, high level language that emphasises the composition of systems and naturally captures the essentials of quantum mechanics. One crucial feature that will be exploited here is the encoding of complementary observables and corresponding phase shifts. Reasoning is done by rewriting diagrams, i.e.\ locally replacing some part of a diagram, according to a few simple rules. Diagrams are defined by their topology only; the number of inputs and outputs and the way they are connected. This exemplifies the `flow' of information. 

Entanglement, described by Einstein as ``spooky action at a distance'', is a key resource in many quantum protocols, like quantum teleportation and quantum cryptography. Yet entanglement makes protocols presented in Dirac notation difficult to verify. In the author's MSc thesis \cite{ik} an alternative presentation in the \zxcalculus of quantum protocols involving the GHZ or the $W$-state is considered. Over 25 different protocols are discussed, such as teleportation, superdense coding (SDC), quantum key distribution (QKD) and quantum direct communication. Moreover, the author has embarked on the task to apply quantum category theory on quantum security; earlier works did not touch anything but BB84 \cite{Bennett_Brassard_1984} in \cite{AbramskyCoecke}. Two of the protocols in \cite{ik}, SDC and QKD, will be discussed in this paper. SDC was first introduced by Bennett and Wiesner in \cite{PhysRevLett.69.2881}. In an SDC protocol a number of classical bits is transfered by transferring fewer qubits. In this paper  SDC with GHZ as described in \cite{superdenseGHZ, multidenseteleport} will be presented in the \zxcalculus. Then a QKD protocol making use of the $W$-state will be presented. In QKD protocols a key is shared with two or more people in such a way that they can only retrieve the key if they work together. The first quantum key distribution protocol was proposed by Bennett and Brassard in \cite{Bennett_Brassard_1984} in 1984 and is generally referred to as the BB84 protocol. After that many other protocols have been proposed \cite{PhysRevLett.67.661,PhysRevLett.68.557, PhysRevLett.68.3121, PhysRevLett.84.4729, PhysRevLett.94.230504, secretschemes1, improvsecretschemes, secretschemes,secret1, secret2}. In this paper QKD as described in \cite{comW} will be discussed. 

To show that a protocol ``works'', the following definitions will be used.
\begin{definition}
 A quantum protocol consists of two parts, the {\bf set of instructions} and the {\bf desired behaviour}. The set of instructions are the things to be done to achieve the desired behaviour, i.e the goal of the protocol.
\end{definition}

\begin{definition}
 A quantum protocol is considered to be {\bf correct} or {\bf valid} if the set of instructions leads to the desired behaviour. 
\end{definition}

This paper is organised as follows. First the \zxcalculus will be introduced in Sec.\ \ref{Sec:redgreen}. Then SDC with GHZ will be presented in Sec.\ \ref{Sec:GHZDC}, after which QKD with $W$ will be discussed in Sec. \ref{comw}. Finally, some concluding remarks will be made in Sec. \ref{conclusions}.


\section{\label{Sec:redgreen} The Red/Green Calculus}

In this section the red/green calculus or the \zxcalculus is introduced. For a more thorough and complete presentation see \cite{Aleks:2009, entanglement, ik} or \cite{intro3}\footnotemark.  
\footnotetext{Pictures from this paper are included with permission of the author.}

\subsection{Components of the \zxcalculus}
The \zxcalculus consists of components joined by wires, like an electrical circuit. Its components are the following:
\begin{enumerate}
 \item  Z-vertices (green dots), labeled by an angle $\alpha \in [0,2 \pi)$, called the phase, with any number of inputs and or outputs, zero included.
\item X-vertices (red dots), complementary to the Z-vertices, labeled by a phase, also with any number of inputs, including none.
\item H-vertices (yellow squares labeled with an H), which represent Hadamard gates. They have exactly one input and one output.
\item$\sqrt{D}$-vertices (black diamonds), which represent scalars. These generally don't have any inputs or outputs.
\end{enumerate}
The Hilbert space interpretation of these components is as as follows:

  \begin{center}
\qquad 
 $\ninlinegraphic{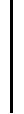} =  
      \left(\begin{array}{cc}
          1&0\\0&1
        \end{array}\right)$  
\qquad \quad 
  $\inlinegraphic{3.2em}{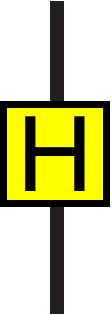} =
      \frac{1}{\sqrt{2}}\left(
        \begin{array}{cc}
          1 &1 \\ 1 & -1
        \end{array}
      \right)$ 
\qquad\quad
      $ \ninlinegraphic{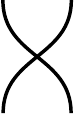} = 
      \left(\begin{array}{cccc}
          1&0&0&0\\0&0&1&0\\0&1&0&0\\0&0&0&1
        \end{array}\right)$
\\[2em] 
   $\ninlinegraphic{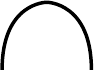}=\ket{00}+\ket{11}$
\qquad \quad      
      $\underbrace{\overbrace{\ninlinegraphic{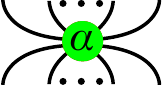}}^n}_m::
      \left\{
      \begin{array}{ccr}
        \overbrace{\ket{0\,\ldots \,0}}^n 
        &\!\! \mapsto\!\! 
        &\ \ \ \ \, \overbrace{\ket{0\,\ldots\, 0}}^m\\
        \ket{1\,\ldots \,1} &\!\! \mapsto\!\! & e^{i\alpha}\ket{1\,\ldots\, 1}\\
        \mbox{others}&\!\! \mapsto\!\! & 0\hspace{1.0cm}
      \end{array}
    \right.$
      \\[1.6em]
\qquad 
  $\ninlinegraphic{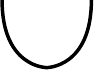}=\bra{00}+\bra{11}$
 \qquad \quad
    \ $\underbrace{\overbrace{\ninlinegraphic{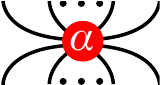}}^n}_m::
    \left\{
    \begin{array}{ccr}
      \overbrace{\ket{+\ldots +}}^n 
      &\!\! \mapsto\!\! 
      &\ \ \ \ \, \overbrace{\ket{+\ldots +}}^m\\
      \ket{-\ldots -} &\!\! \mapsto\!\! & e^{i\alpha}\ket{-\ldots -}\\
      \mbox{others}&\!\! \mapsto\!\! & 0\hspace{1.1cm}
    \end{array}
  \right.$
        \\[1.6em]

$\ninlinegraphic{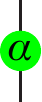}= Z^1_1(\alpha) =\left(
  \begin{array}{cc}
    1 & 0 \\ 0 & e^{i \alpha}
  \end{array}\right)$ \qquad \quad
  $\ninlinegraphic{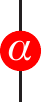}= X^1_1(\alpha) =
  e^{-i{\alpha/2}}
  \left(
    \begin{array}{cc}
      \cos\frac{\alpha}{2} & i\sin\frac{\alpha}{2}\\
      i\sin\frac{\alpha}{2} & \cos\frac{\alpha}{2}
    \end{array}
  \right)$\\[1.6em] 
\qquad\quad
      $\inlinegraphic{4.5mm}{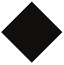} = \sqrt{2}$\qquad \quad
$\ninlinegraphic{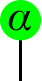}=\ket{0} + e^{i\alpha}\ket{1}$ 
\qquad \quad$\ninlinegraphic{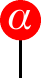}=\cos{\frac{\alpha}{2}}\ket{0} + i\sin{\frac{\alpha}{2}}\ket{1}$
  \end{center}


\begin{figure}[htb]
  \centering
  \framebox[\linewidth]{
  \begin{minipage}{0.98\linewidth}
    \begin{tabular}{cc}
      \it{``Only the topology matters''}
      & \eruleT \\[1em]
      \inlinegraphic{5.13em}{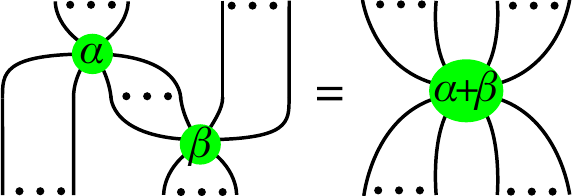} 
      \quad
       \inlinegraphic{5.13em}{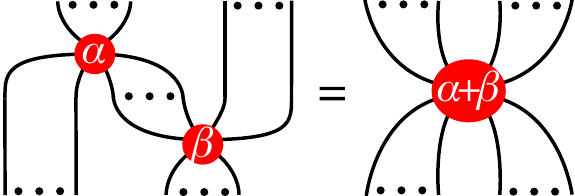}$\!\!\!\!\!\!$
      & \eruleSi \\[3em]
      \inlinegraphic{2.7em}{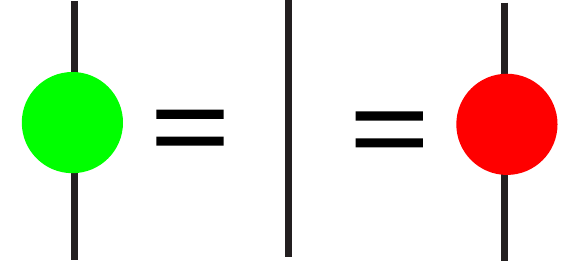} 
      \qquad
      \inlinegraphic{3.50em}{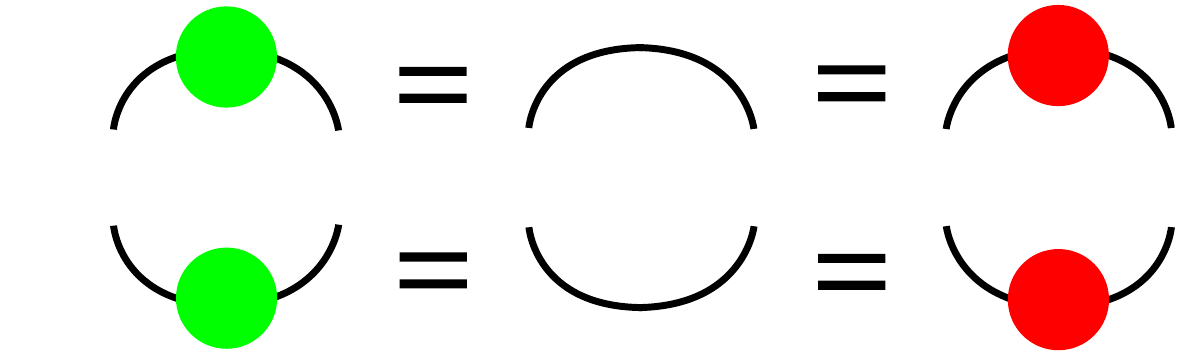}
      & \eruleSii      \\[2em]
      \inlinegraphic{8.5mm}{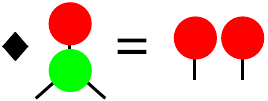}
      \quad
      \inlinegraphic{8.5mm}{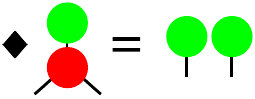}
      \quad
       \eruleBi 
      \quad\ \
      \inlinegraphic{4.0em}{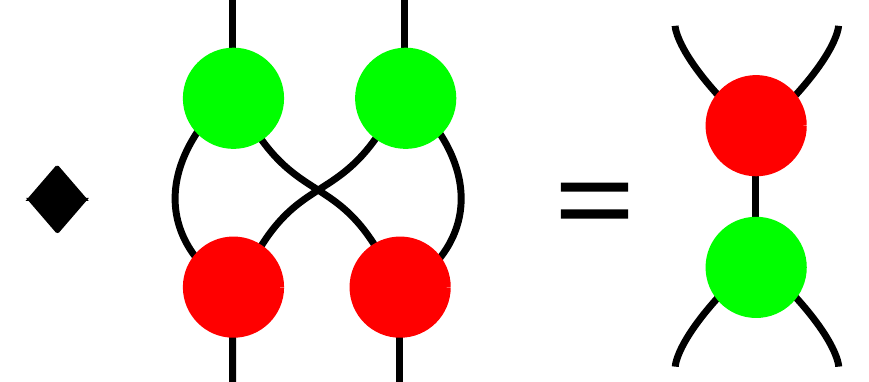}\quad\  
      & \eruleBii \\[2em]
     $\!$ \inlinegraphic{11.5mm}{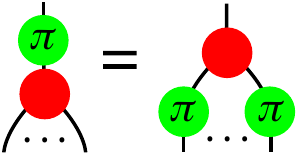}
      \  \, $\!\!$
      \inlinegraphic{11.5mm}{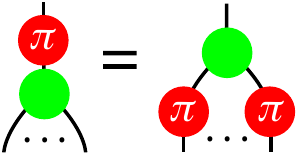}
      \  $\!$
       \eruleKi
       \quad\ \ \
      \inlinegraphic{10.2mm}{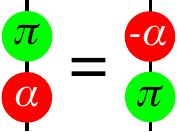}     
      \ \ 
      \inlinegraphic{10.2mm}{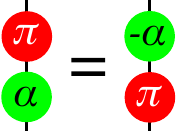}$\!\!$
      & \eruleKii \\[2em]
        $\inlinegraphic{8mm}{zx-gens-redspider-alpha} = \ 
        \inlinegraphic{12.6mm}{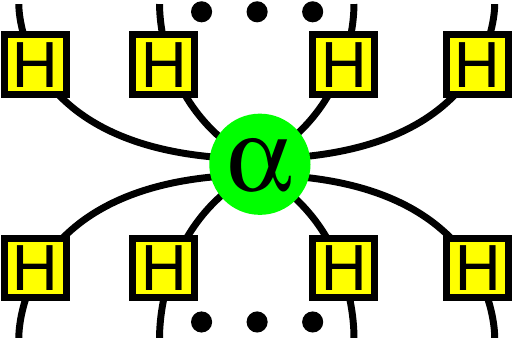}$
      & \eruleC \\[2em]
  \  \qquad  \qquad  $\inlinegraphic{8mm}{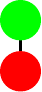} =\ 
  \inlinegraphic{3mm}{black-diamond}$
      \qquad \qquad\quad \eruleDi \qquad \quad
      \inlinegraphic{2.0em}{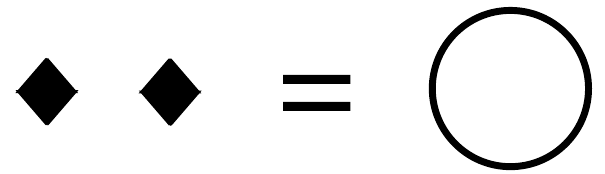}
      & \eruleDii \\[1em]
    \end{tabular}
  \end{minipage}}

\caption{Basic Rules for the \zxcalculus}  \label{basicrules}
\end{figure}
\subsection{Quantomatic}
From now on all the pictures (except Fig. \ref{basicrules}) are made with \verb$quantomatic$ \cite{quantomatic}, a software tool for reasoning with the \zxcalculus, developed jointly at Oxford, Cambridge, Google and Edinburgh. Different rule sets can be loaded into \verb$quantomatic$. One can then input a graph and see what rewrites are possible within the loaded rules. Download and installation instructions can be found at  \url{http://sites.google.com/site/quantomatic/home}. Using \verb$quantomatic$ to produce the pictures confirms that all the rewrites are valid, but \verb$quantomatis$ does not (yet) provide a way to determine the order in which these rewrites should be applied; \verb$quantomatic$'s normalising tool often halts or gets into an infinite loop.

In \verb$quantomatic$ red and green dots are displayed as red and green dots, with their phase in a box of corresponding colour underneath the vertex. Hadamard gates are displayed as yellow boxes with an $H$ in them. Finally, inputs and outputs are displayed as grey boxes with a number in it, called boundary vertices. \verb$Quantomatic$ does not distinguish between inputs and outputs, so they look the same and are numbered with the same counter. E.g. a diagram with one input and one output would have one boundary vertex with the number zero and one with the number one in it, as in the diagrammatic representation of the Hopf law in the next section. Which output or input gets which number is not important and depends solely on the order of input in \verb$quantomatic$. 

\subsection{Basic Rules}
Derivations in the Red/Green-calculus are done mostly by a few simple rules, outlined in Fig.\  \ref{basicrules}. Note that rule \ruleT \text{ }does not mean that the topology is always preserved; other rules might change this completely. Informally \ruleT \text{ }can be seen as ``yanking'' the wires, making sure the number of inputs and outputs is preserved and the way they are connected.

An addition to these rules has been made by Kissinger in \cite{Aleks:2009} and Coecke and Edwards in \cite{entanglement}. These two related rules that are called $\ket{0}$- and $\ket{1}$- supplementarity were found by solving a matrix equation in \cite{Aleks:2009} and by means of the underlying  algebraic structure in \cite{entanglement}.

\begin{center}
 $\ninlinegraphic[0.5]{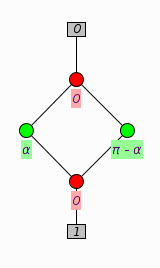}=\ninlinegraphic[0.5]{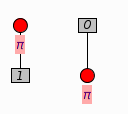} \,, \ \ 
  \ninlinegraphic[0.5]{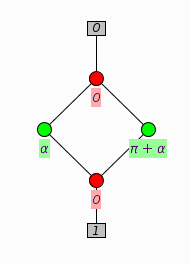}=\ninlinegraphic[0.5]{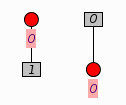}\qquad \eruleE\,.$ 
\end{center}

From now on scalars will be left out for sake of simplicity. Note also that \verb$quantomatic$ does not include scalars in the rewrites. 

A useful derivation from the basic rules is the Hopf law:

\begin{lemma}[Hopf law]
  \[
  \ninlinegraphic[0.5]{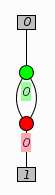} = \
  \ninlinegraphic[0.5]{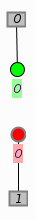}    \quad \ \ \ ({\bf
    B}')\]

\end{lemma}
\begin{proof}
\begin{align}
\nonumber
  \ninlinegraphic[0.5]{hopf-proof-informal-1}  &\stackrel{\eruleT}{=} 
  \ninlinegraphic[0.5]{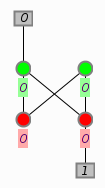}  \stackrel{\eruleS}{=} 
  \ninlinegraphic[0.5]{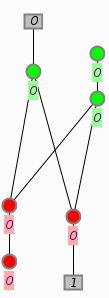}  \stackrel{\eruleBii}{=} 
  \ninlinegraphic[0.5]{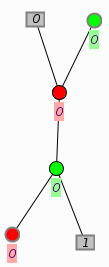}   \stackrel{\eruleBi}{=} 
  \ninlinegraphic[0.5]{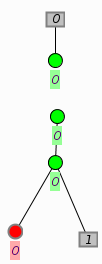}  \\&\stackrel{\eruleSi}{=} 
  \ninlinegraphic[0.5]{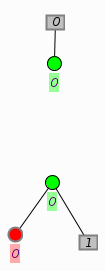}  \stackrel{\eruleSii}{=} 
  \ninlinegraphic[0.5]{hopf-proof-informal-7}
\end{align}
\end{proof}

Another useful derivation, used later on is the following:

\begin{lemma}\label{ZX-derive-A}
 
  \[
  \ninlinegraphic[0.5]{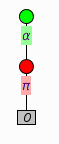} = \
  \ninlinegraphic[0.5]{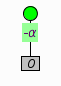}    \quad \ \ \ \eruleA \]

\end{lemma}
\begin{proof}
\begin{equation}
  \ninlinegraphic[0.5]{Arule1}  \stackrel{\eruleS}{=} 
  \ninlinegraphic[0.5]{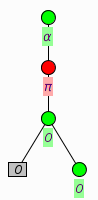}  \stackrel{\eruleKii}{=} 
  \ninlinegraphic[0.5]{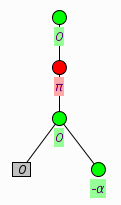}  \stackrel{\eruleKi}{=} 
  \ninlinegraphic[0.5]{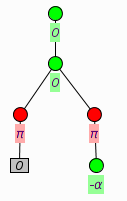}   \stackrel{\eruleS}{=} 
  \ninlinegraphic[0.5]{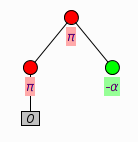}  \stackrel{\eruleS}{=} 
  \ninlinegraphic[0.5]{Arule7}.
\end{equation}
\end{proof}

\subsection{Measurements into the $x$- and $z$-basis}

The graphical representation of measurements or ``effects'' into the $x$- and $z$-basis can be derived by the Hilbert space interpretation of points. $\ket{+}$ ($\ket{x^+}$) and $\ket{-}$ ($\ket{x^-}$) and $\ket{0}$ ($\ket{z^+}$) and $\ket{1}$($\ket{z^-}$) are represented as
\begin{equation}
  \ninlinegraphic[0.5]{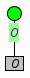}=\ket{+} \, \ \
  \ninlinegraphic[0.5]{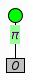}=\ket{-} \,\ \ 
  \ninlinegraphic[0.5]{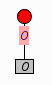}=\\\ket{0} \, \ \
  \ninlinegraphic[0.5]{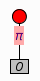}=\ket{1} \,.
    \end{equation}

\subsection{The GHZ state}
The GHZ state is one of the only two SLOCC-inequivalent classes of tripartite entanglement \cite{inequiventanglement}. SLOCC-inequivalent means inequivalent under Stochastic Local Operations and Classical Communication, i.e.\ one cannot be turned into the other by means of stochastic local operations (unitaries and or measurements) and classical communication. GHZ is defined as $\ket{GHZ}= \frac{1}{\sqrt{2}}(\ket{000}+ \ket{111})$, or as the map
\[ GHZ :: \begin{cases}
      \ket{0} \mapsto \ket{00}\\
      \ket{1} \mapsto \ket{11}.
    \end{cases}
\]
Graphically it is represented as \cite{intro3}
\begin{equation}
 \label{eq:GHZstate}
\ninlinegraphic[0.5]{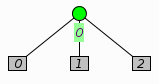}.
\end{equation}
Plugging $\ket{0}$ and $\ket{1}$ gives
\begin{equation}
  \ninlinegraphic[0.5]{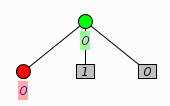} \stackrel{\eruleBi}{=}
    \ninlinegraphic[0.5]{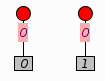}= \ket{00}\end{equation} {and}
\begin{equation}
    \ninlinegraphic[0.5]{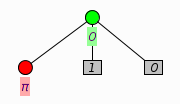} \stackrel{\eruleBi/ \eruleKi}{=}
    \ninlinegraphic[0.5]{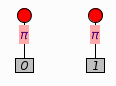}= \ket{11},
\end{equation}
as required.

\subsection{The $W$-state}
The class of $W$-states is SLOCC-inequivalent to the class of GHZ states \cite{inequiventanglement}. The $W$-state is defined as
\[\ket{W}= \frac{1}{\sqrt{3}}(\ket{001}+ \ket{010} + \ket{100}),
\]
or as the map \cite{Aleks:2009}
\[ W :: \begin{cases}
      \ket{0} \mapsto \ket{01} + \ket{10}\\
      \ket{1} \mapsto \ket{00}.
    \end{cases}
\]
Considering this map, the $W$-state can be graphically represented as \cite{Aleks:2009}
\begin{equation}
 \label{eq:Wstate}
\ninlinegraphic[0.5]{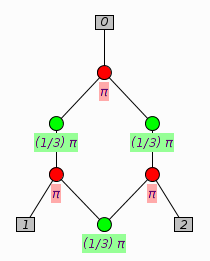}.
\end{equation}
This graphical representation shows the robustness of the $W$-state; due to the pairwise entanglement, after tracing out one of the qubit, there is still the possibility of bipartite entanglement, as opposed to the GHZ state, which is fully separable when any of the qubits is traced out. Plugging $\ket{0}$ in Eq.\  \ref{eq:Wstate} gives
\begin{align}
\nonumber
  \ninlinegraphic[0.5]{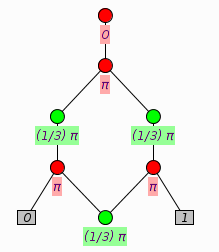} &\stackrel{\eruleSi}{=}
    \ninlinegraphic[0.5]{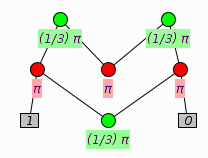}   \quad\stackrel{\eruleKii, \eruleS}{=}
    \ninlinegraphic[0.5]{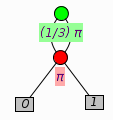} \end{align}\begin{align} &\stackrel{\eruleBp}{=}
    \label{eq:w30} \ninlinegraphic[0.5]{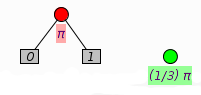}= \ket{01} + \ket{10},
\end{align}
as expected. Plugging $\ket{1}$ in Eq.\  \ref{eq:Wstate} gives
\begin{align}
\nonumber
  \ninlinegraphic[0.5]{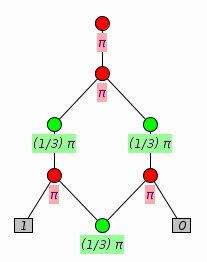} &\stackrel{\eruleSi}{=}
    \ninlinegraphic[0.5]{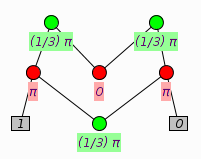} \stackrel{\eruleS}{=}
    \ninlinegraphic[0.5]{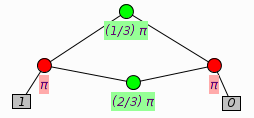}  \\&\stackrel{\eruleE}{=}\label{lemma:plug1W}
    \ninlinegraphic[0.5]{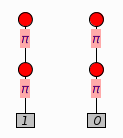}\stackrel{\eruleSi}{=} \ninlinegraphic[0.5]{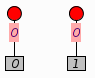}= \ket{00},
\end{align}
as required.

\subsection{The Pauli Matrices}
By the Hilbert space interpretation, one can obtain the representation of the Pauli-Z and Pauli-X matrices in \zxcalculus:
\begin{center}
$\sigma_z=
\ninlinegraphic[0.5]{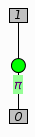} = \left(
   \begin{array}{rr}
     1&0\\ 0&-1
   \end{array}
\right)
\qquad \sigma_x=
\ninlinegraphic[0.5]{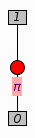} = \left(
   \begin{array}{rr}
     0&1\\ \ 1&\ 0
   \end{array}
\right)\;.$
\end{center}

Once the Pauli-$X$ and $Z$ matrices are known, it is easily deduced that $i \sigma_y = \sigma_z \times \sigma_x= \sigma_z \circ \sigma_x$ and $- i \sigma_y= \sigma_x \times \sigma_z = \sigma_x \circ \sigma_z$ are compositions of $\sigma_x$ and $\sigma_z$. 

\begin{center}
$ i \sigma_y = \ninlinegraphic[0.5]{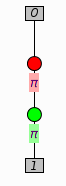} = \left(
   \begin{array}{rr}
     0&1\\ -1&0
   \end{array}
\right)
\qquad
-i \sigma_y = \ninlinegraphic[0.5]{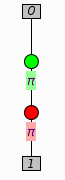} = \left(
   \begin{array}{rr}
     0&-1\\ \ 1&\ 0
   \end{array}
\right)\;.$
\end{center}

\subsection{Controlled Not gate}
Another useful tool is the Controlled Not gate. The Controlled Not gate is defined as follows: 
\begin{equation}\label{eq:infCNOT}
\textbf{CNOT}=
\ninlinegraphic[0.5]{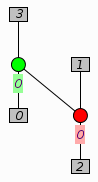} = \ninlinegraphic[0.5]{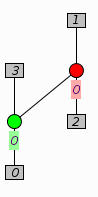} =  
\left(\begin{array}{cccc}
1 & 0 & 0 & 0\\
0 & 1 & 0 & 0\\
0 & 0 & 0 & 1\\
0 & 0 & 1 & 0
\end{array}\right).
\end{equation}


\section{\label{Sec:GHZDC} Superdense Coding with GHZ} 
As with Bell states, it is possible to transfer an amount of classical bits by transferring fewer qubits by means of different states in the GHZ class. When states in the GHZ class are used for super dense coding, two qubits need to be transfered in order to transfer three classical bits \cite{superdenseGHZ}. This section is divided up in five subsections. First the steps of the protocol will be explained. Then the eight different states in the GHZ class will be presented. In the last three subsections it will be shown how, through measurement in the GHZ basis, these eight different states can be translated into three classical bits, proving the validity of the protocol. 

\subsection{Super Dense Coding with GHZ protocol}
Provided Alice and Bob share $\ket{GHZ}$, such that the first qubit belongs to Bob and the other two qubits belong to Alice, the following protocol describes superdense coding with GHZ as in \cite{superdenseGHZ, multidenseteleport}:
\begin{enumerate}
 \item Alice applies a combination of $I, \sigma_x, i\sigma_y$ and $\sigma_z$ on both her qubits, encoding one of eight distinguishable states in the GHZ class. 
 \item Alice transfers both her qubits to Bob.
 \item Bob measures all three qubits in the GHZ basis, retrieving the state Alice encoded. 
 \item Bob translates the retrieved state to three classical bits. 
\end{enumerate}

\subsection{Different GHZ states}
There are eight different states in the GHZ class. One can go from one to the other by performing unitary single particle operations on two of the three particles. These unitary operations are $I, \sigma_x, i\sigma_y$ and $\sigma_z$. Though there are 16 different combinations of these operators on two qubits, only half of them generate distinguishable states \cite{Wang2005} as can be seen by comparing the graphical representations in Table \ref{T:dc} and \ref{T:dc2} in Appendix \ref{ap:GHZ}. In this section we will work with states in the standard form 
\begin{equation}
 \ket{GHZ_{+ij}}= \frac{1}{\sqrt{2}} (\ket{0ij} + \ket{1 \overline{i} \overline{j}}),
\end{equation}
as in \cite{MultiQSSSEPR, superdenseGHZ}, where $i,j \in \{0,1\}$, $\overline{i}=1-i$ and $\overline{j}=1-j$.

\subsection{Encoding a GHZ measurement outcome into classical bits}
The encoding is as in \cite{measurement}. After measurement, an output qubit is encoded as 0 if it is $\ket{0}$ and as 1 if it is$\ket{1}$. Every GHZ state gives three output bits.  If the first output bit is 0, then there is an odd number of $\ket{+}$ in the basis, otherwise an even number. If the second output bit is 0, then the first two bits in the GHZ class state are the same, otherwise they are different. And finally, if the last output bit is 0, the last two bits in the GHZ class state are the same. This encoding is displayed in Table \ref{T:dc} in Appendix \ref{ap:GHZ}.

\subsection{Measurement into the GHZ basis}
The circuit to measure into the GHZ basis is \cite{measurement}
\begin{equation}\label{eq:GHZmeasurement}
  \ninlinegraphic[0.5]{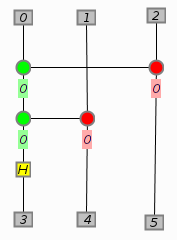},
\end{equation}
which gives three qubits in the $z$-basis. Plugging states in the $z$-basis, we obtain for $\alpha, \beta, \gamma \in \{0, \pi\}$
\begin{equation}
\label{eq:GHZmeasurementoutcome}
 \ninlinegraphic[0.5]{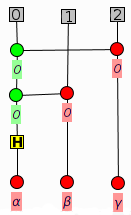}\stackrel{\eruleC}{=}\ninlinegraphic[0.5]{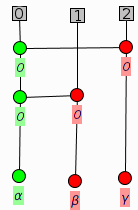}\stackrel{\eruleS}{=}\ninlinegraphic[0.5]{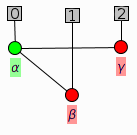},
\end{equation}
which results in one of eight states in the GHZ class upside down by Table \ref{T:dc} and \ref{T:dc2} in Appendix \ref{ap:GHZ} if values for $\alpha$, $\beta$ and $\gamma$ are set.

\subsection{Validity of the protocol}
\begin{lemma}\label{lemma:SDCGHZ}
 The encoding in combination with the GHZ measurement circuit in \cite{measurement} makes for a valid Superdense Coding protocol \footnotemark.
\footnotetext{Note that this encoding is different from \cite{superdenseGHZ, multidenseteleport}.}

\end{lemma}
\begin{proof}
Let $\ket{0}=0$ and $\ket{1}=1$ in classical bits. What needs to be shown is that for all eight states in the GHZ class, their measurement outcome is equal to their binary representation in Table \ref{T:dc} in Appendix \ref{ap:GHZ}.
\begin{align}
\nonumber
 \ninlinegraphic[0.5]{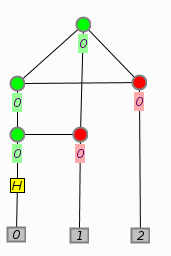} &\stackrel{\eruleS}{=}\ninlinegraphic[0.5]{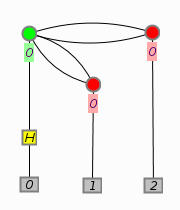} \stackrel{\eruleBp}{=}\ninlinegraphic[0.5]{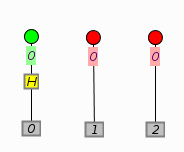} \\\label{eq:SDC1}&\stackrel{\eruleC}{=}\ninlinegraphic[0.5]{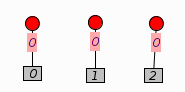} = \ket{000} = 000=0
\end{align}
\begin{align}
\nonumber
 \ninlinegraphic[0.5]{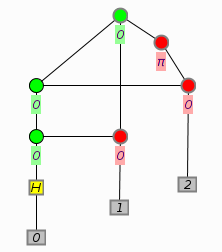} &\stackrel{\eruleS}{=}\ninlinegraphic[0.5]{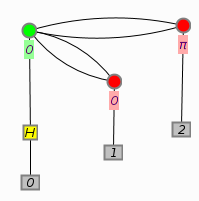} \stackrel{\eruleBp / \eruleC}{=}\ninlinegraphic[0.5]{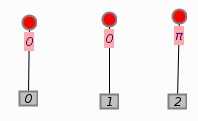} \\&= \ket{001} = 001=1\end{align}\begin{align}\nonumber
 \ninlinegraphic[0.5]{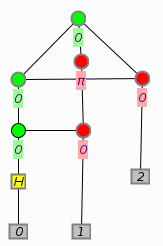} &\stackrel{\eruleS}{=}\ninlinegraphic[0.5]{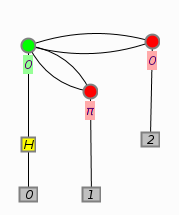} \stackrel{\eruleBp / \eruleC}{=}\ninlinegraphic[0.5]{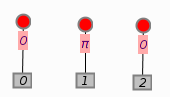} \\& = \ket{010} = 010=2\end{align}\begin{align}
\nonumber
 \ninlinegraphic[0.5]{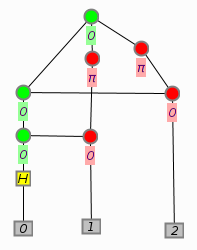} &\stackrel{\eruleS}{=}\ninlinegraphic[0.5]{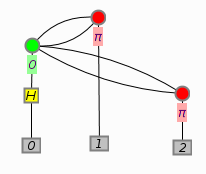} \stackrel{\eruleBp / \eruleC}{=}\ninlinegraphic[0.5]{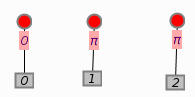} \\&= \ket{011} = 011=3
\end{align}
\begin{align}
\nonumber
\ninlinegraphic[0.5]{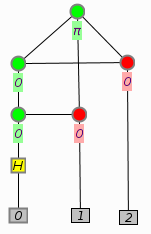} &\stackrel{\eruleS}{=}\ninlinegraphic[0.5]{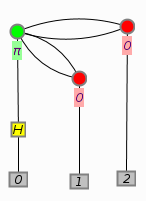} \stackrel{\eruleBp / \eruleC}{=}\ninlinegraphic[0.5]{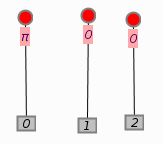} \\
& = \ket{100} = 100=4\end{align}\begin{align}
\nonumber
 \ninlinegraphic[0.5]{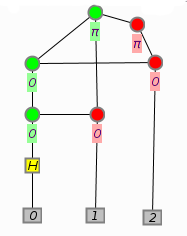} &\stackrel{\eruleS}{=}\ninlinegraphic[0.5]{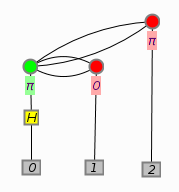} \stackrel{\eruleBp / \eruleC}{=}\ninlinegraphic[0.5]{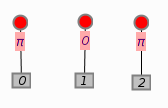} \\
 & = \ket{101} = 101=5\end{align}\begin{align}\nonumber
 \ninlinegraphic[0.5]{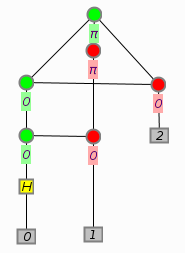} &\stackrel{\eruleS}{=}\ninlinegraphic[0.5]{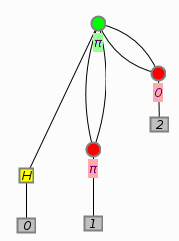} \stackrel{\eruleBp / \eruleC}{=}\ninlinegraphic[0.5]{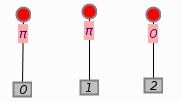} \\& = \ket{110} = 110=6
\end{align}
\begin{align}
\nonumber
  \ninlinegraphic[0.5]{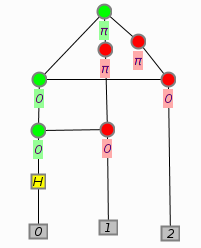} &\stackrel{\eruleS}{=}\ninlinegraphic[0.5]{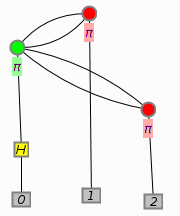} \stackrel{\eruleBp / \eruleC}{=}\ninlinegraphic[0.5]{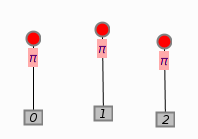} \\
\label{eq:SDC2}
& = \ket{111} = 111=7.
\end{align}
Eq.\  \ref{eq:SDC1}-\ref{eq:SDC2} imply the validity of the protocol. 
\end{proof}

\subsection{Superdense coding with $N$-GHZ}\label{sec:Nsuperdense}
In a similar way superdense coding for $N$-GHZ states can be constructed. One of $\{I,\sigma_x, i \sigma_y,  \sigma_z\}$ can be applied on the $N^{\text{th}}$ qubit and one of $\{I, \sigma_x\}$ on qubit $2 - (N-1)$ to encode $2^N$ different states. They can be distinguished with a  measurement like the GHZ basis measurement \cite{measurement}.

\section{Pairwise Quantum Key Distribution with $W_3$}\label{comw}

In this section the Quantum Key Distribution protocol with $W_3$ from \cite{comW} will be explained. It will moreover be shown by means of the \zxcalculus that this protocol works.
\subsection{Pairwise Quantum Key Distribution with $W_3$ protocol}
Alice, Bob and Charlie share a series of $W_3$-states in the Pairwise Quantum Key Distribution with $W_3$ protocol and perform measurements on their qubits in such way that two of them will share a common (classical) key. Assuming they share a series of  $W_3$-states, the protocol can be established as follows:
\begin{enumerate}
 \item All choose at random the $x$- or the $z$-basis to measure their qubit in.
\item Each announces publicly his or her measurement direction.
\item \label{it:security} For security reasons, they randomly choose to announce their measurement outcomes, to check for eavesdropping. If they do, the protocol is restarted.
\item \label{item:basis}If the overall measurement basis is $z-x-x$, $x-z-x$ or $x-x-z$, they continue with the protocol. Otherwise they start over and discard these measurement outcomes.
\item The one who measured along the $z$-axis is the decider. S/he tells the others whether the outcome is $\bra{z^+}$. Otherwise they restart the protocol.
\item \label{item:outcome}The other two now know that they have the same outcome, i.e. they share a bit now. 
\item Repeat the protocol until the desired amount of key bits are obtained.
\item Use the information from step \ref{it:security} to check for eavesdropping. If eavesdropping is detected, discard the obtained key bits and start a new quantum channel to repeat the protocol.
\end{enumerate}
\begin{figure}[ht]
   \centering
    \includegraphics[width=0.5 \textwidth]{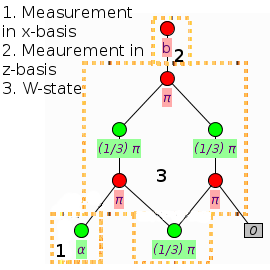}
\caption{\label{fig:Wcom} Graphical representation of the set of instructions of the Pairwise Quantum Key Distribution with $W_3$ protocol. $\alpha, b \in {0, \pi}$}
\end{figure}
\begin{lemma}
 Fig.\  \ref{fig:Wcom} is the graphical representation of the set of the instructions of the Pairwise Quantum Key Distribution with $W_3$ protocol, when the first two measurements are in the $z$- and $x$-basis.
\end{lemma}
\begin{proof}
 Box 1 is a measurement in the $x$-basis, box 2 is a measurement into the $z$-basis. Box 3 is the $W_3$-state. 
\end{proof}
\begin{lemma}
 If any one of them gets the outcome $\bra{z^-}$ the entanglement will be broken and the outcomes for the other two will be $\bra{z^+}$.
\end{lemma}
\begin{proof}
 Setting $b=\pi$, then by Eq\ \ref{lemma:plug1W}.
\end{proof}
\begin{lemma}
\label{lemma:Wcorrect}
 When there is a proper overall measuring basis and the decider has the outcome $\bra{z^+}$, the other two always have the same result.
\end{lemma}
\begin{proof}
Let $\alpha \in \{0, \pi\}$ and set $b=0$, then
 \begin{align}
\nonumber
 \ninlinegraphic[0.5]{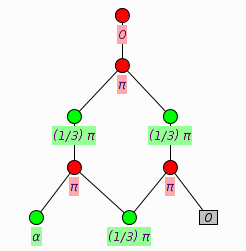} &\stackrel{\eruleSi, \eruleKi}{=}
    \ninlinegraphic[0.5]{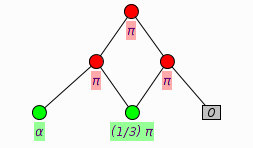} \stackrel{\eruleSi}{=}
    \ninlinegraphic[0.5]{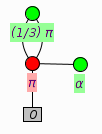} \\&\stackrel{\eruleBp}{=}
\ninlinegraphic[0.5]{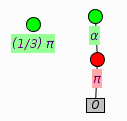} \stackrel{\eruleA}{=}
    \ninlinegraphic[0.5]{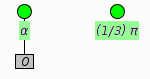}. 
\end{align} 
\end{proof}
\begin{corollary}
 The Pairwise Quantum Key Distribution with $W_3$ protocol is correct.
\end{corollary}

\section{Conclusions}\label{conclusions}
In this paper two of the 25 plus quantum protocols in \cite{ik} have been presented in the \zxcalculus. It has been shown that they are both valid protocols, by means of easy manipulations of their diagrammatic representation. Producing the diagrammatic representation of a protocol is a matter of ``gluing'' together basic constructions, such as CNOT gates, Bell states and measurements. The more of these basic constructions are known, the easier this becomes. Many of these basic constructions can be found in the author's MSc thesis, along with examples in which they are used \cite{ik}.  After this representation has been found, the correctness of a protocol can be shown by a few simple manipulations of the diagrams. From this it can be concluded that the red/green calculus provides a clear and intuitive way to represent and check quantum protocols. 

Though not explicitily discussed here, it turns out that it is difficult to check the security of many protocols by means of the \zxcalculus. This is because detection of eavesdropping or disturbance often depends on the probability of obtaining illegal measurement outcomes after the quantum channel is interfered with. Although it is easy to plug different measurement outcomes in the diagrammatic notation, it is not always possible to calculate what combinations of measurement outcomes are possible. Moreover one cannot calculate the probability of different combinations of measurement outcomes from the diagrammatic notation, because in the implementation of the the \zxcalculus, \verb$quantomatic$ \cite{quantomatic}, scalars are ignored. 

One can conclude from this, that a presentation in either Dirac notation or the \zxcalculus is not optimal. Instead, a combined approach should be taken. 
The \zxcalculus provides a simple and intuitive way to understand the workings of the protocol and Dirac notation is helpful to check the security of protocols. 

\bibliography{EPTCS.bib}
\bibliographystyle{eptcs}
\newpage
\appendix
\section{The Graphical representation of GHZ class States} \label{ap:GHZ}

\def\padding{0.5mm}
\def\scaling{0.37}

\begin{table}[h!]
\begin{center}
 \begin{tabular}{|c|c|c|c|c|} \hline
 \# & Binary & Standard Form (SF) & Unitaries &Graphical Representation \\ \hline\hline
0 & 000 & $\frac{1}{\sqrt{2}}(\ket{000} + \ket{111})$ & $I \otimes I$& $\growbox{\padding}{\ninlinegraphic[\scaling]{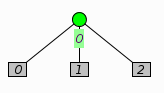}}$\\ \hline
1 & 001 & $\frac{1}{\sqrt{2}}(\ket{001} + \ket{110})$ & $I \otimes \sigma_x$&$\growbox{\padding}{\ninlinegraphic[\scaling]{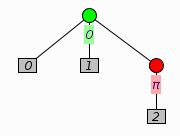}}$\\ \hline
2 & 010 & $\frac{1}{\sqrt{2}}(\ket{010} + \ket{101})$ & $\sigma_x\otimes I$&$\growbox{\padding}{\ninlinegraphic[\scaling]{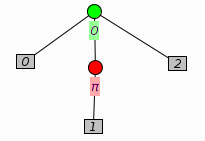}}$\\ \hline
3 & 011 & $\frac{1}{\sqrt{2}}(\ket{011} + \ket{100})$ & $\sigma_x \otimes \sigma_x$&$\growbox{\padding}{\ninlinegraphic[\scaling]{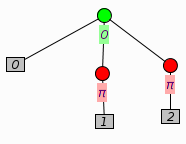}}$\\ \hline
4 & 100 & $\frac{1}{\sqrt{2}}(\ket{000} - \ket{111})$ & $\sigma_z \otimes I$&$\growbox{\padding}{\ninlinegraphic[\scaling]{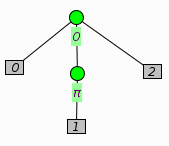}} \stackrel{\eruleS}{=}
    \growbox{\padding}{\ninlinegraphic[\scaling]{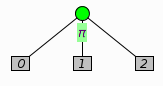}}$\\ \hline
5 & 101 & $\frac{1}{\sqrt{2}}(\ket{001} - \ket{110})$ & $\sigma_z \otimes \sigma_x$&$\growbox{\padding}{\ninlinegraphic[\scaling]{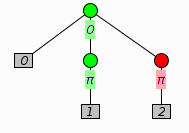}} \stackrel{\eruleS}{=}
    \growbox{\padding}{\ninlinegraphic[\scaling]{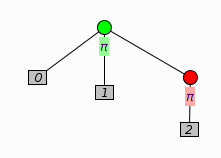}}$\\ \hline
6 & 110 & $\frac{1}{\sqrt{2}}(\ket{010} - \ket{101})$ & $i\sigma_y \otimes I$&$\growbox{\padding}{\ninlinegraphic[\scaling]{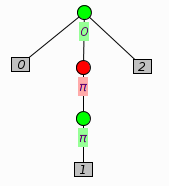}} \stackrel{\eruleKii}{=}
    \growbox{\padding}{\ninlinegraphic[\scaling]{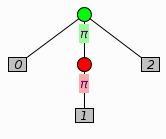}}$\\ \hline
7 & 111 & $\frac{1}{\sqrt{2}}(\ket{011} - \ket{100})$ & $ i \sigma_y \otimes \sigma_x$&$\growbox{\padding}{\ninlinegraphic[\scaling]{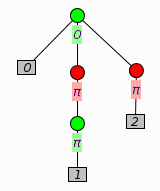}} \stackrel{\eruleS}{=}
    \growbox{\padding}{\ninlinegraphic[\scaling]{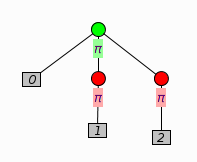}}$\\ \hline
 \end{tabular}
\end{center}

\caption{\label{T:dc}This table shows the eight different states in the GHZ class, their binary presentation, the unitaries that should be applied to the second and the third qubit to obtain this state from $\ket{GHZ}$ and finally their graphical representation.}
\end{table}

\begin{table}[h!]
\begin{center}
 \begin{tabular}{|c|c|c|c|c|}
\hline
  \# & Binary & Alternative Form (AF) & Unitaries &Graphical Representation  \\ \hline\hline
0 & 000 &$\frac{1}{\sqrt{2}}(\ket{111} + \ket{000})$&$\sigma_z \otimes \sigma_z$&$\growbox{\padding}{\ninlinegraphic[\scaling]{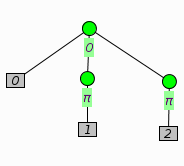}} \stackrel{\eruleS}{=}
    \growbox{\padding}{\ninlinegraphic[\scaling]{GHZ}}$\\ \hline
1 & 001 &$\frac{1}{\sqrt{2}}(\ket{110} + \ket{001})$&$\sigma_z \otimes i\sigma_y$& $\growbox{\padding}{\ninlinegraphic[\scaling]{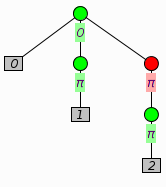}} \stackrel{\eruleS,\eruleKii}{=}
    \growbox{\padding}{\ninlinegraphic[\scaling]{GHZ2}}$\\ \hline
2 & 010 &$\frac{1}{\sqrt{2}}(\ket{101} + \ket{010})$&$ i \sigma_y \otimes \sigma_z$& $\growbox{\padding}{\ninlinegraphic[\scaling]{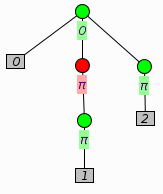}} \stackrel{\eruleS, \eruleKii}{=}
    \growbox{\padding}{\ninlinegraphic[\scaling]{GHZ3}}$\\ \hline
3 & 011 &$\frac{1}{\sqrt{2}}(\ket{100} + \ket{011})$&$i\sigma_y \otimes i \sigma_y$& $\growbox{\padding}{\ninlinegraphic[\scaling]{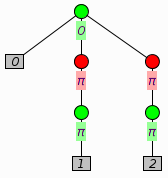}} \stackrel{\eruleS, \eruleKii}{=}
    \growbox{\padding}{\ninlinegraphic[\scaling]{GHZ4}}$\\ \hline
4 & 100 & $\frac{1}{\sqrt{2}}(\ket{111} - \ket{000})$ & $I \otimes \sigma_z$&$\growbox{\padding}{\ninlinegraphic[\scaling]{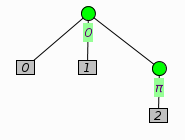}} \stackrel{\eruleS}{=}
    \growbox{\padding}{\ninlinegraphic[\scaling]{GHZ52}}$\\ \hline
5 & 101 & $\frac{1}{\sqrt{2}}(\ket{100} - \ket{001})$ & $I \otimes i \sigma_y$&$\growbox{\padding}{\ninlinegraphic[\scaling]{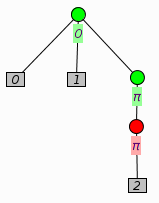}} \stackrel{\eruleS,\eruleKii}{=}
    \growbox{\padding}{\ninlinegraphic[\scaling]{GHZ62}}$\\ \hline
6 & 110 & $\frac{1}{\sqrt{2}}(\ket{101} - \ket{010})$ & $\sigma_x \otimes \sigma_z$&$\growbox{\padding}{\ninlinegraphic[\scaling]{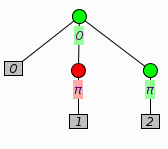}} \stackrel{\eruleS}{=}
    \growbox{\padding}{\ninlinegraphic[\scaling]{GHZ72}}$\\ \hline
7 & 111 & $\frac{1}{\sqrt{2}}(\ket{100} - \ket{011})$ & $ \sigma_x \otimes i\sigma_y$&$\growbox{\padding}{\ninlinegraphic[\scaling]{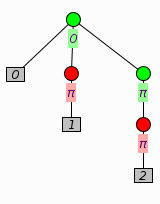}} \stackrel{\eruleS,\eruleKii}{=}
    \growbox{\padding}{\ninlinegraphic[\scaling]{GHZ82}}$\\ \hline
 \end{tabular}
\end{center}

\caption{\label{T:dc2}This table shows the remaining states in the GHZ class, their binary presentation, the unitaries that should be applied to the second and the third qubit to obtain this state from $\ket{GHZ}$ and finally their graphical representation.}
\end{table}

\end{document}